\documentclass[runningheads,a4paper]{llncs}

\title{ Site-Directed
Insertion: Decision Problems, Maximality and 
Minimality\thanks{Cho and Han were supported by the Basic 
	Science Research Program
through NRF (2015R1D1A1A01060097) and
the International Research \& Development Program of NRF (2017K1A3A1A12024971).
Salomaa and Smith were supported by Natural Sciences and
Engineering Research Council of Canada Grant OGP0147224.}}
\titlerunning{Site-Directed Insertion
}
\author{Da-Jung Cho\inst{1} \and Yo-Sub
Han\inst{1} \and  
Kai Salomaa\inst{2} \and Taylor J. Smith\inst{2}}
\institute{Department of Computer Science, Yonsei University\\ 
50, Yonsei-Ro, Seodaemun-Gu, Seoul 120-749, Republic of Korea\\
{\tt \{ dajungcho, emmous \}@yonsei.ac.kr } \\
\and
School of Computing, Queen's University \\
Kingston, Ontario K7L 2N8, Canada\\
{\tt \{ ksalomaa, tsmith \}@cs.queensu.ca}\\
}
	\authorrunning{D.-J. Cho, Y.-S. Han, K. Salomaa, T.J. Smith}
\usepackage{amssymb,amsmath}
\usepackage{wasysym}
\usepackage{epsfig}
\usepackage{graphicx}
\usepackage{latexsym}
\usepackage{enumerate}
\usepackage{algorithm}
\usepackage{algorithmic}
\usepackage{clrscode}
\usepackage{bbold}
\usepackage{tikz}

\usepackage{latexsym}

\definecolor{aocolour}{rgb}{0.7,0.8,1}
\definecolor{mkcolour}{rgb}{1,0.9,0.7}

\newcommand{\seepage}[1]{\marginpar{\scriptsize (p.~\pageref{#1})}}

\begin{document}

\renewcommand{\labelenumi}{(\roman{enumi})}

\newtheorem{open}{Open problem}

\newcommand{\inv}{\mathbb{INV}}
\newcommand{\PI}{\textit{pseudo-inversion}}
\newcommand{\mpi}{\mathbb{PI}}
\newcommand{\sdi}{\mathbb{SDI}}

\newcommand{\rfig}[1]{Fig.~\ref{#1}} 
\newcommand{\rlem}[1]{Lemma~\ref{#1}}
\newcommand{\rdef}[1]{Definition~\ref{#1}}
\newcommand{\rthm}[1]{Theorem~\ref{#1}}
\newcommand{\rpro}[1]{Proposition~\ref{#1}}
\newcommand{\rcor}[1]{Corollary~\ref{#1}} 
\newcommand{\rse}[1]{Section~\ref{#1}}
\newcommand{\robs}[1]{Observation~\ref{#1}}
\newcommand{\req}[1]{Equation~(\ref{#1})}
\newcommand{\ralg}[1]{Algorithm~\ref{#1}}
\newcommand{\rtab}[1]{Table~\ref{#1}}
\newcommand{\rsec}[1]{Section~\ref{#1}}
\newcommand{\rapp}[1]{Appendix~\ref{#1}}


\maketitle

\begin{abstract}
	Site-directed insertion is an overlapping insertion
	operation that can be viewed as analogous to the
	overlap assembly or chop operations that concatenate
	strings by overlapping a   suffix and a prefix of the
	argument strings.
We consider decision problems and language equations involving
site-directed insertion. By relying on the tools provided
by {\em semantic shuffle on trajectories\/} we show that one variable
equations involving site-directed insertion and regular constants
can be solved. We consider also maximal and minimal variants
of the site-directed insertion operation.
\end{abstract}

\section{Introduction}

Site-directed mutagenesis is one of the most important
techniques for generating mutations on specific sites
of DNA using polymerase chain reaction (PCR) based
methods~\cite{Reikofski1992}. The algorithmic applications
of mutagenesis have been considered e.g. by
Franco and Manca~\cite{Franco2011}.
Contextual insertion/deletion systems in the study of molecular
computing have been used, e.g. by
Kari and Thierrin~\cite{Kari1996}, Daley et al.~\cite{Daley1999}
and Enaganti et al.~\cite{Enaganti2015}.

Site-directed insertion (SDI) of a string $y$ into a string $x$ involves
matching an outfix of $y$ with a substring of $x$ and inserting the
``middle part'' of $y$ not belonging to the outfix into $x$.
Site-directed insertion has earlier been considered under the name
{\em outfix-guided insertion\/}~\cite{Cho2017}.
The operation is an overlapping
variant of the insertion operation in the same sense as
the overlap assembly, a.k.a. chop operation, is a variant
of string 
concatenation~\cite{Csuhaj2007,Enaganti2017,Holzer2012,Holzer2017}.

The maximal (respectively, minimal) SDI of a string $y$ into a
string $x$ requires that, at the chosen location of $x$, the operation
matches a maximal (respectively, minimal) outfix of $y$ with
a substring of $x$. This is analogous to the
maximal and minimal chop operations
 studied by Holzer et al.~\cite{Holzer2017}.

Site-directed insertion can be represented as a 
{\em semantic shuffle on trajectories\/} (SST).
Shuffle on trajectories was introduced by Mateescu et al.~\cite{Mateescu1998}
and the extension to SST is due to Domaratzki~\cite{Domaratzki2004}.
Further extensions of the shuffle-on-trajectories operation have
been studied by Domaratzki et al.~\cite{Domaratzki2006}.

	Here we study decision problems and language equations involving
site-directed insertion and its maximal and minimal variants.
The representation of SDI as a semantic shuffle on a regular
set of trajectories
gives regularity preserving left- and right-inverses of the operation.
By the general results 
of Kari~\cite{Kari1994} on the decidability of equations, translated for
SST by Domaratzki~\cite{Domaratzki2004}, this makes it possible
to decide linear equations involving SDI where the constants are
regular languages.

The maximal and minimal SDI operations do not, in general, preserve
regularity. This means that the operations cannot be
represented by SST~\cite{Domaratzki2004} (on a regular set
of trajectories) and the above tools
are not available to deal with language equations.
We   show that for maximal and minimal SDI
certain independence properties related to coding
applications~\cite{Jurgensen1997} can be decided in a polynomial
time. The decidability
of whether a regular language is closed under max/min SDI remains open.

In the last section we give a tight bound for the nondeterministic
state complexity of alphabetic SDI,
where the matching outfix must consist of a prefix and suffix
of length exactly one. An upper bound for the state complexity of
the general site-directed insertion is known but it remains open
whether the bound is optimal.

\section{Preliminaries}
\label{kaksi}

We assume the reader to be familiar with the
basics of  finite
automata, regular languages and context-free
languages~\cite{Shallit2009}. Here we briefly
recall some notation.

Let $\Sigma$ be an alphabet and $w \in \Sigma^*$. If we can
write $w = x y z$ we say that the pair $(x, z)$ is an
{\em outfix\/} of $w$. The outfix $(x, z)$ is a 
{\em nontrivial outfix\/} of $w$
if $x \neq \varepsilon$ and $z \neq \varepsilon$.
For $L \subseteq \Sigma^*$, $\overline{L} = \Sigma^*
- L$ is the complement of $L$.

A {\em nondeterministic finite automaton\/} (NFA) is a tuple
$A = (\Sigma, Q, \delta, q_0, F)$ where $\Sigma$ is
the input alphabet, $Q$ is the finite set of states,
$\delta \colon Q \times \Sigma \rightarrow 2^Q$ is
the transition function, $q_0 \in Q$ is the
initial state and $F \subseteq Q$ is the set of final states.
In the usual way $\delta$ is extended as a function
$Q \times \Sigma^* \rightarrow 2^Q$ and the {\em language accepted
by\/} $A$ is $L(A) = \{ w \in \Sigma^* \mid \delta(q_0, w) \cap F
\neq \emptyset \}$. The automaton $A$ is a {\em deterministic finite
automaton\/} (DFA) if $|\delta(q, a)| \leq 1$ for all $q \in Q$
and $a \in \Sigma$.
It is well known that the deterministic and nondeterministic
finite automata recognize the class of {\em regular languages.}

The so called fooling set lemma gives a technique for establishing
lower bounds for the size of NFAs:
\begin{lemma}[Birget 1992 \cite{Birget1992}]
	\label{latta23}
Let~$L\subseteq \Sigma^*$ be a regular language. 
Suppose that there exists a set $P=\{(x_i,w_i)\mid 1\leq i\leq n\}$ 
of pairs of strings such that: (i)
$x_iw_i\in L$ for $1\leq i\leq n$, and, (ii)
 if $i\neq j$, then $x_iw_j\not \in L$ or $x_jw_i\not \in L$ 
 for $1\leq i,j\leq n$.
Then, any minimal NFA for $L$ has at least $n$ states.
\end{lemma}

Finally we recall some notions concerning  operations on languages
and language equations.
Let $\odot$ be a binary operation on languages, and $L$,
$R$ are languages over an alphabet $\Sigma$. 
\begin{enumerate}
\item The language
 $L$ is {\em closed under $\odot$\/} if
$L \odot L \subseteq L$.
\item  The language $L$ is {\em $\odot$-free\/} with respect to $R$ if
$L \odot R = \emptyset$.
\item The language $L$ is {\em $\odot$-independent\/} with
	respect to $R$ if $(L \odot \Sigma^+) \cap R = \emptyset$.
\item A {\em solution\/} for an equation $X \odot L = R$ (respectively, 
$L \odot X = R$)  is a language $S \subseteq \Sigma^*$ such
that $S \odot L = R$ (respectively, $L \odot S = R$).
\end{enumerate}
The $\odot$-freeness and independence properties
 can be related to coding applications,
where it might be desirable that we cannot produce new strings
by applying an operation, such as site-directed insertion,
to strings of the language. Domaratzki~\cite{Domaratzki2004b} defines
trajectory-based codes analogously with (iii).
As we will see,  languages that
are site-directed insertion independent with respect to themselves
have a
definition closely resembling outfix-codes of index one \cite{Jurgensen1997}.

\section{Site-Directed Insertion}
\label{kolme}

The site-directed insertion is a partially overlapping insertion
operation analogously as the overlap-assembly (or self-assembly)
\cite{Csuhaj2007,Enaganti2017} models an overlapping concatenation
of strings. The overlapping concatenation operation is also  called
the  chop operation \cite{Holzer2017}.

The {\em site-directed insertion\/} (SDI) of a string $y$ into a string $x$
is defined as
$$
x \stackrel{\rm sdi}{\leftarrow} y = \{ x_1 u z v x_2 \mid x = x_1 u v x_2,
\; y = uzv, \; u \neq \epsilon, v \neq \epsilon \}.
$$
The above definition requires that the pair $(u, v)$ is a nontrivial outfix
of the string $y$ and $uv$ is a substring of $x$.
If $y = uzv$ is inserted into $x$ by matching the outfix with
a substring $uv$ of $x$, 
we say that $(u,v)$ is an \emph{insertion guide} for the operation.
Note that a previous paper
	\cite{Cho2017} uses the name ``outfix-guided insertion''
for the same operation. 

The site-directed insertion operation
is extended in the usual way for languages by setting
$$
L_1 \stackrel{\rm sdi}{\leftarrow} L_2 = \bigcup_{w_i \in L_i, i = 1, 2}
w_1 \stackrel{\rm sdi}{\leftarrow} w_2.
$$
We recall that regular languages are closed under
site-directed insertion.

\begin{proposition}[\cite{Cho2017}]
\label{patta31}
If $A$ and $B$ are NFAs with $m$ and $n$ states, respectively,
the language $L(A) \stackrel{\rm sdi}{\leftarrow} L(B)$ has
an NFA with $3mn + 2m$ states.
\end{proposition}

A simpler form of the overlap-assembly
operation requires the overlapping part
of the strings to consist of a single letter. This operation is called
``chop'' by Holzer and
Jacobi~\cite{Holzer2012} but the later definition of the chop-operation
\cite{Holzer2017} coincides with general overlap-assembly \cite{Enaganti2017}.
Analogously we define {\em alphabetic site-directed insertion\/} by
requiring that the overlapping prefix and suffix of the inserted string
each consist of a single letter.

The \emph{alphabetic site-directed insertion}
 of a string $y$ into a string $x$
is
$$
x \stackrel{\rm a-sdi}{\leftarrow} y = \{ x_1 a z b x_2 \mid x = x_1 ab x_2,
\; y = azb, \; a, b \in \Sigma, \; x_1, x_2, z \in \Sigma^* \}.
$$

Note that  the alphabetic site-directed insertion 
will have different closure properties
than the standard site-directed insertion. 
For example,
it is not difficult to see that the context-free languages are closed
under alphabetic site-directed insertion, while the context-free
languages are not closed under general site-directed insertion~\cite{Cho2017}.


\subsection{Decision problems}

For a regular language $L$, it is
decidable whether $L$ is closed under site-directed insertion.
The algorithm relies on the construction of Proposition~\ref{patta31}
and operates in polynomial time when $L$ is specified by a 
DFA~\cite{Cho2017}.
Deciding whether a context-free language is closed under site-directed
insertion is undecidable~\cite{Cho2017}.

A language  $L$ is $\stackrel{\rm sdi}{\leftarrow}$-free,
or SDI-free,
with respect to $R$
 if
no string of $R$ can be site-directed inserted into
a string of $L$, that is, if $L \stackrel{\rm sdi}{\leftarrow} R 
= \emptyset$.  The language $L$  is 
SDI-independent  with respect to $R$ if site-directed inserting
a non-empty string into $L$ cannot produce a string of $R$.
Note that $L$ being SDI-independent with respect to itself resembles
the notion of $L$ being an outfix-code of index one~\cite{Jurgensen1997}
with the difference that we require the outfix to be nontrivial.
For example, $\{ ab, b \}$ is SDI-independent but it is not
an outfix-code of index one.

\begin{theorem}
	\label{tatta20}
For NFAs $A$ and $B$ we can decide in polynomial time
whether 
\begin{enumerate}
	\item $L(A)$ is SDI-free
(or SDI-independent) with respect to $L(B)$.
	\item $L(A)$ is alphabetic SDI-free
(or alphabetic SDI-independent) with respect to $L(B)$.
\end{enumerate}
\end{theorem}

For context-free languages deciding SDI-freeness and
SDI-independence is undecidable.

\begin{proposition}
	\label{patta61}
	For context-free languages $L_1$ and $L_2$ it is
	undecidable whether 
	\begin{enumerate}
		\item$L_1$ is SDI-free
	with respect to $L_2$,
\item $L_1$ is SDI-independent with respect to $L_2$.
	\end{enumerate}
\end{proposition}

For dealing with language equations we express the
site-directed insertion operation as a 
{\em semantic shuffle on a set of trajectories\/} (SST)
due to Domaratzki~\cite{Domaratzki2004}.
The semantic shuffle extends the (syntactic) {\em shuffle
on trajectories\/} originally defined by Mateescu
et al.~\cite{Mateescu1998}.
We use a simplified
definition of SST 
that does not allow {\em content restriction\/}~\cite{Domaratzki2004}. 

The {\em trajectory alphabet\/} is $\Gamma = \{ 0, 1, \sigma \}$
and a trajectory is a string over $\Gamma$.
The semantic shuffle of $x, y \in \Sigma^*$ on  a trajectory
$t \in \Gamma^*$, denoted by $x \pitchfork_t y$, is defined as follows.

If $x = y = \varepsilon$, then $x \pitchfork_t y =  \varepsilon$
if $t = \varepsilon$ and is undefined otherwise. If $x = ax'$,
$a \in \Sigma$, $y =
\varepsilon$ and  $t = c t'$, $c \in \Gamma$, then
$$
x \pitchfork_t \varepsilon = \begin{cases}
	a (x' \pitchfork_{t'} \varepsilon) \mbox{ if } c = 0,\\
	\emptyset, \mbox{otherwise.}
\end{cases}
$$
If $x = \varepsilon$, $y = b y'$, $b \in \Sigma$, and 
$t = c t'$, $c \in \Gamma$, then
$$
\varepsilon \pitchfork_t y =
\begin{cases}
	b (\varepsilon \pitchfork_{t'} y') \mbox{ if } c = 1,\\
	\emptyset, \mbox{otherwise.}
\end{cases}
$$
In the case where all the strings are nonempty, for
$x = ax'$, $y = by'$, $a, b \in \Sigma$, and $t = ct'$,
$c \in \Gamma$, we define
$$
x \pitchfork_t y = \begin{cases}
	a(x' \pitchfork_{t'} y) \mbox{ if } c = 0,\\
	b(x \pitchfork_{t'} y') \mbox{ if } c = 1,\\
	a(x' \pitchfork_{t'} y') \mbox{ if } a = b \mbox{ and } c = \sigma,\\
	\emptyset, \mbox{ otherwise.}
\end{cases}
$$
Intuitively, the trajectory $t$ is a sequence of instructions that
guide the shuffle of strings $x$ and $y$: 0 selects the next symbol of
$x$, 1 the next symbol of $y$ (these are as in the original definition
of syntactic shuffle \cite{Mateescu1998}) and $\sigma$ represents synchronized
insertion where the next symbols of the argument strings must coincide.

For $x, y \in \Sigma^*$ and $t \in \Gamma^*$, $x \pitchfork_t y$ either
consists
of a unique string or is undefined. For $T \subseteq \Gamma^*$,
$x \pitchfork_T y = \bigcup_{t \in T} x \pitchfork_t y$ and the operation
is extended in the natural way for languages over $\Sigma$.

Directly from the definition it follows that the SDI and
alphabetic SDI operations can be represented as semantic shuffle
on a regular set of trajectories.

\begin{proposition}
	\label{patta67}
	Let $T_{\rm sdi} = 0^* \sigma^+ 1^* \sigma^+ 0^*$ and
	$T_{\rm a-sdi} = 0^* \sigma 1^* \sigma 0^*$. Then, for any languages
$L_1$ and $L_2$,
$$
L_1 \stackrel{\rm sdi}{\leftarrow} L_2 = 
L_1 \pitchfork_{T_{\rm sdi}} L_2, \mbox{ and, }
L_1 \stackrel{\rm a-sdi}{\leftarrow} L_2 = 
L_1 \pitchfork_{T_{a-sdi}} L_2.
$$
\end{proposition}

Now using the strong decidability results of
Domaratzki~\cite{Domaratzki2004} 
we can effectively decide linear language equations
involving site-directed insertion where the constants are
regular languages. The representation of SDI using SST guarantees the
existence of regularity preserving left- and right-inverse
of the operation. This makes it possible to use the results
of Kari~\cite{Kari1994} to decide existence
of solutions to linear equations where the constants are
regular languages. The maximal solutions to the equations
are represented using {\em semantic deletion along
trajectories\/} \cite{Domaratzki2004}. For the deletion
operation we consider a trajectory
alphabet $\Delta = \{ i, d, \sigma \}$. Intuitively,
a trajectory $t \in \Delta^*$ guides the deletion
of a string $y$ from $x$ as follows:
a symbol $i$ (insertion) indicates that we output the next symbol of $x$,
a symbol $d$ (deletion) indicates that the next symbol of $y$ must
match the next symbol of $x$ and nothing is produced in the output
and a symbol $\sigma$ (synchronization) indicates that the next symbols
of $x$ and $y$ must match and this symbol
is placed in the output.
 The result of deleting $y$ from $x$ along trajectory $t$
is denoted $x \leadsto_t y$ and the operation is extended in
the natural way for sets of trajectories and for languages.

We can express the
left- and right-inverse 
(as defined in~\cite{Kari1994,Domaratzki2004}) of SDI
 using semantic deletion along trajectories, and these relations
are used to express solutions for linear language equations.
Given a binary operation $\Diamond$ on strings, let
$\Diamond^{\rm rev}$ be the operation defined by 
$x \Diamond^{\rm rev} y = y \Diamond x$
for all $x, y \in \Sigma^*$.
Using Theorems~6.4 and~6.5 of \cite{Domaratzki2004} we obtain:

\begin{theorem}
	\label{tatta62}
	Let $L, R \subseteq \Sigma^*$ be regular languages.
	Then for each of the following equations it is
	decidable whether a solution exists:
	(a) $X \stackrel{\rm sdi}{\leftarrow} L = R$,
	(b) $L \stackrel{\rm sdi}{\leftarrow} X = R$,
	(c) $X \stackrel{\rm a-sdi}{\leftarrow} L = R$,
	(d) $L \stackrel{\rm a-sdi}{\leftarrow} X = R$.

	Define $T_1 = i^* \sigma^+ d^* \sigma^+ i^*$,
	$T_1^a =  i^* \sigma d^* \sigma i^*$,
	$T_2 = d^* \sigma^+ i^* \sigma^+ d^*$, and
	$T_2^a =  d^* \sigma i^* \sigma d^*$.
		If a solution 
		 exists, a superset of all
		solutions is, respectively, for the different
cases:
		(a) $S_a = \overline{\overline{R} \leadsto_{T_1} L}$,
		(b)  $S_b = \overline{L (\leadsto_{T_2})^{\rm rev}
 \overline{R}}$,
		(c) $S_c = \overline{\overline{R} \leadsto_{T_1^a} L}$,
		(d)  $S_d = \overline{L (\leadsto_{T_2^a})^{\rm rev}
 \overline{R}}$.
\end{theorem}

The above result does not give a polynomial time decision algorithm,
even in the case where the languages $L$ and $R$ are given by
DFA's. Semantic shuffle on  and deletion
along regular sets of trajectories preserve regularity but
the operations are inherently nondeterministic and complementation
blows up the size of an NFA. Note that deleting an individual string $y$
from a string $x$ along trajectory $t$ is deterministic, but the
automaton construction for the result of the operation
on two DFA languages is nondeterministic.
An explicit construction of an NFA
for the  syntactic shuffle of two regular languages is given
in \cite{Mateescu1998}.

The known trajectory based methods for two variable
equations~\cite{Domaratzki2004} 
do not allow the trajectories to use the synchronizing symbol $\sigma$
that is needed to represent the overlap of SDI. However, if we are just interested
to know whether a solution exists (as opposed to finding
maximal solutions), it is easy to verify that
an equation $X_1  \stackrel{\rm sdi}{\leftarrow} X_2 = R$ 
has a solution if and only if all strings of $R$ have length at least two.

\section{Maximal and minimal site-directed insertion}
\label{nelja}

Holzer et al. \cite{Holzer2017} define two deterministic variants
of the chop operation. The max-chop  (respectively, min-chop)
of strings $x$ and $y$ chooses
the non-empty suffix of  $x$ overlapping with $y$ to be as
long (respectively, as short) as possible. 

By a \emph{maximal site-directed insertion\/} of string $y$ into
a string $x$ we mean, roughly speaking, an insertion where neither
the overlapping prefix nor the overlapping suffix can be properly
extended. The operation is not deterministic because $y$ could be
inserted in different positions in $x$. At a specific position in $x$,
a string $y$ can be maximally (respectively, minimally) inserted
in at most one way.

Formally, the {\em maximal site-directed insertion\/} (max-SDI)
of a  string $y$ into string
$x$  is defined as follows:
\begin{eqnarray*}
	x  \stackrel{\rm max-sdi}{\leftarrow}  y  & = \{ & x_1 u z v x_2 
		 \mid  x = x_1 u v x_2,
\; y = uzv, \; u \neq \epsilon, v \neq \epsilon, 
\mbox{ and }  \\
& & \mbox{there exist no suffix } 
\mbox{$x_1'$ of $x_1$  and prefix $x_2'$ of $x_2$ }
\mbox{such that } \\
& & 
\mbox{ 
$x_1' x_2' \neq \epsilon$ and $y = x_1' u z' v x_2'$, $z' \in \Sigma^*$ } \; \}
\end{eqnarray*}
Equivalently the maximal SDI of $x$ and $y$ is
\begin{eqnarray*}
	x  \stackrel{\rm max-sdi}{\leftarrow} y  & =   \{ & x_1 u z v x_2 
	\mid x = x_1 u v x_2,
\; y = uzv, \; u \neq \epsilon \neq  v , 
\mbox{ no   suffix of $x_1 u$} 
\\
& & \mbox{of length greater than $|u|$ is a prefix of $uz$ and no prefix}\\
& & \mbox{
of $vx_2$ of length greater than $|v|$ is a suffix of $zv$
 } \; \}.
\end{eqnarray*}

\begin{figure}
\epsfxsize=12.5cm
\epsfbox{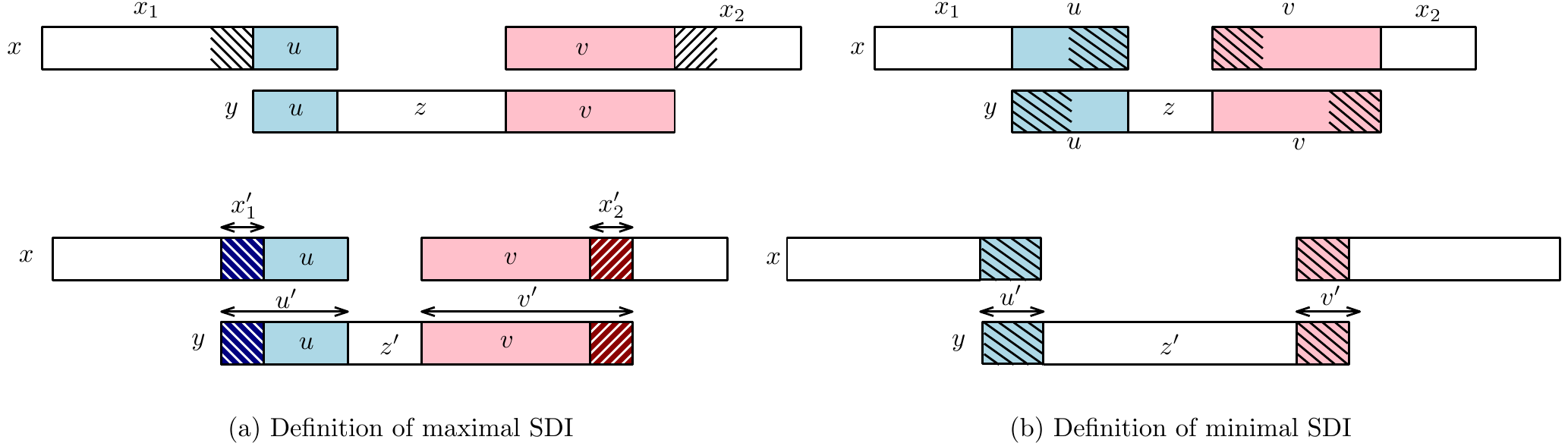}
\caption{
\label{fatta1}
Insertion of $y$ into $x$ depicted at top left is not maximal when $x$ and
$y$ have decompositions as depicted at bottom left.
}
\end{figure}

In particular, if $x$ and $y$ are unary strings with $|x| \geq |y| \geq 2$, 
then
$
x \stackrel{\rm max-sdi}{\leftarrow} y = x
$
because the maximal overlapping outfix consists always of the
entire string $y$. If $|x| \geq 2$ and $|y| > |x|$, then
$x \stackrel{\rm max-sdi}{\leftarrow} y = y$. If $|x| < 2$
or $|y| < 2$, the operation
is undefined.

\begin{example}
	{\rm  Consider alphabet $\Sigma = \{ a, b, c \}$. Now
		$$
		ababab \stackrel{\rm max-sdi}{\leftarrow} acbab = \{ 
		acbabab, abacbab, ababacbab \}
		$$
	For example also the string
	$abacbabab$  is obtained by site-directed inserting
	$y = acbab$ into $x = ababab$. In this operation the prefix $a$
of $y$ is matched with the 3rd symbol of $x$ and the suffix $b$ of
$y$ is matched with the 4th symbol of $x$. However, this operation
does not satisfy the maximality condition because after the 3rd symbol
of $x$ we can match a longer suffix of $y$.
} 
\end{example}

The {\em minimal site-directed insertion\/}
(min-SDI) operation is defined as follows:
\begin{eqnarray*}
	x  \stackrel{\rm min-sdi}{\leftarrow} y  & = \{ & x_1 u z v x_2 
	\mid x = x_1 u v x_2,
\; y = uzv, \; u \neq \epsilon, v \neq \epsilon, \\
& & 
\mbox{no proper nonempty
suffix of $u$ is a prefix of $u$, and } \\
& & \mbox{no proper nonempty prefix
of $v$ is a suffix of $v$ } \}.
\end{eqnarray*}
Note that in the definition of min-SDI, $u$ and $v$ are
{\em unbordered words.}
Figure~\ref{fatta1}~(b) illustrates the defininition of minimal SDI.
The alphabetic SDI can be viewed as an ``extreme'' case of minimal
SDI: if the first and last letter of $y$ coincide with a substring
of $x$ of length two, then the alphabetic and minimal site-directed insertion
of $y$ in that position coincide.

If $x$ and $y$ are unary strings with $|x|, |y| \geq 2$, then
$	x \stackrel{\rm min-sdi}{\leftarrow} y$ is the
unary string of length $|y| + |x| - 2$ and the operation
is undefined for $|x| < 2$ or $|y| < 2$.

Note that while the maximal or minimal SDI is
considerably more restricted than the unrestricted SDI operation,
if a string $y$ can be site-directed inserted to a string $x$,
it can be also maximally or minimally inserted at the
same position. The result
of an alphabetic insertion is always a minimal insertion.
These observations are formalized in the next lemma.

\begin{lemma}
	\label{latta21}
	Let $x, y \in \Sigma^*$.
	\begin{enumerate}
		\item 
			$ x \stackrel{\rm max-sdi}{\leftarrow} y
			\subseteq 
			x \stackrel{\rm sdi}{\leftarrow} y$ and
			$ x \stackrel{\rm a-sdi}{\leftarrow} y
	\subseteq x \stackrel{\rm min-sdi}{\leftarrow} y
			\subseteq 
			x \stackrel{\rm sdi}{\leftarrow} y$.
		\item
			$x \stackrel{\rm sdi}{\leftarrow} y \neq \emptyset$
			iff 
	$x \stackrel{\rm max-sdi}{\leftarrow} y \neq \emptyset$
	iff
		$x \stackrel{\rm min-sdi}{\leftarrow} y \neq \emptyset$.
	\item It is possible that
$x \stackrel{\rm min-sdi}{\leftarrow} y \neq \emptyset$ and
$x \stackrel{\rm a-sdi}{\leftarrow} y = \emptyset$.
	\end{enumerate}
\end{lemma}

Since the max-chop and min-chop operations do not preserve regularity
\cite{Holzer2017},
it can be expected that the same holds for maximal and minimal
SDI. The proof of the following proposition is inspired by
Theorem~3 of \cite{Holzer2017}.

\begin{proposition}
	\label{patta1}
	The maximal and minimal site-directed insertion do not
	preserve regularity.
\end{proposition}

\begin{proof}
Let $\Sigma = \{ a, b, \$, \% \}$ and choose
$$
L_1 = ba^+ ba^+ \$, \;\;\; L_2 = ba^+ ba^+ \% \$
$$
We claim that
$$
(L_1  \stackrel{\rm max-sdi}{\leftarrow} L_2) \cap (ba^+)^3 \% \$ =
\{ b a^m b a^n b a^k \% \$ \mid m \neq n \mbox{ or } k < n,
\; m, n, k \geq 1 \} 
$$
We denote the right side of the equation by $L_{\rm result}$ which
is clearly nonregular.
Since the strings of $L_2$ contain the marker \% that does not occur
in strings of $L_1$,
when inserting a string $y \in L_2$ into a string of $L_1$ the
overlapping suffix of $y$ must consist exactly of the last symbol
\$. Consider $x = b a^i b a^j \$ \in L_1$ and
$y = b a^r b a^s \% \$ \in L_2$. In order for the resulting string
to have three symbols $b$, a prefix of $ba^r$ must overlap with
$b a^j$, that is, $j \leq r$. In order for the overlap to be
maximal we must have $r \neq i$ or $s < j$. These relations
guarantee that the unique string in
$x  \stackrel{\rm max-sdi}{\leftarrow} y$ is in $L_{\rm result}$.

For the converse inclusion we note that, for $m \neq n$ or $k < n$, 
$$
ba^m b a^n ba^k \% \$ \in ba^m ba^n \$ 
\stackrel{\rm max-sdi}{\leftarrow} b a^n ba^k \% \$.
$$

For non-closure under min-SDI we claim that
$$
(L_1  \stackrel{\rm min-sdi}{\leftarrow} L_2) \cap (ba^+)^2 \% \$ =
\{ b a^m b a^n  \% \$ \mid n > m \geq 1 \}
 =^{\rm def} L'_{\rm result}.
$$
Consider $x = b a^i b a^j \$ \in L_1$ and
$y = b a^r b a^s \% \$ \in L_2$. In order for the result of
site-directed insertion of $y$  into $x$
to have two $b$'s, $b a^i b a^j$ must be a prefix of 
$b a^r b a^s$, that is, $i = r$ and $j \leq s$. For the site-directed
insertion to be minimal, no proper non-empty prefix of $b a^i b a^j$
can be its suffix, that is $i < j$. These relations guarantee
that the minimal SDI of $x$ and $y$ is in $L'_{\rm result}$.

Conversely, for $n > m$,
$b a^m b a^n \% \$ \in b a^m b a^n \$
\stackrel{\rm min-sdi}{\leftarrow} ba^m ba^n \% \$ $.
\qed
\end{proof}

In fact, extending the max-chop and min-chop constructions
from Theorem~3 of \cite{Holzer2017} 
it would be possible to show that there exist
regular languages $L_1$ and $L_2$ such that
$L_1 \stackrel{\rm max-sdi}{\leftarrow} L_2$
(or $L_1 \stackrel{\rm min-sdi}{\leftarrow} L_2$) is
not context-free.
The maximal or minimal site-directed insertion of a finite
language into a regular language (and vice versa) is regular.

\begin{proposition}
	\label{patta2}
Let $R$ be a regular language and $L$ a finite language.
Then the languages
$R \stackrel{\rm max-sdi}{\leftarrow} L$,
$R \stackrel{\rm min-sdi}{\leftarrow} L$,
$L \stackrel{\rm max-sdi}{\leftarrow} R$,
and $L \stackrel{\rm min-sdi}{\leftarrow} R$
are regular.
\end{proposition}

\begin{proof}
We show that $R \stackrel{\rm max-sdi}{\leftarrow} L$ is regular.
The other cases are very similar.

Since 
$$R \stackrel{\rm max-sdi}{\leftarrow} L = 
\bigcup_{y \in L} R  \stackrel{\rm max-sdi}{\leftarrow} y
$$
and regular languages are closed under finite union, it is sufficient
to consider the case where $L$ consists of one string $y$.

Let $A$ be an NFA for $R$ and $y \in \Sigma^*$. 
We outline how an NFA $B$ can recognize
$L(A)  \stackrel{\rm max-sdi}{\leftarrow} y$.
On an input $w$, $B$ nondeterministically guesses a decomposition
$w = x_1 y_1 y_2 y_3 x_2$ where $x_1 y_1 y_3 x_2 \in L(A)$,
$y_1 y_2 y_3 = y$ and $y_1, y_3 \neq \varepsilon$. When reading the
prefix $x_1 y_1$, $B$ simulates a computation of $A$ ending
in a state $q$, then skips
the substring $y_2$, and continues simulation of $A$ from state
$q$ on the suffix $y_3 x_2$.

The above checks that the input is in
$L(A) \stackrel{\rm sdi}{\leftarrow} y$ and, additionally, $B$
needs to verify that the insertion is maximal. This is possible
because $B$ is looking for maximal insertions of the one
fixed string $y$.

(i) When processing the prefix $x_1$, the state of $B$ remembers
the last $|y| - 1$ symbols scanned. When the computation
nondeterministically guesses the substrings $y_1$, $y_2$, $y_3$,
it can then check that for no nonempty suffix $x_1'$ of $x_1$,
$x_1' y_1$ is a prefix of $y_1 y_2$. If this condition does
not hold, the corresponding transition is undefined.

(ii) Similarly, when processing the 
(nondeterministically selected) suffix $x_2$ of the input,
$B$ remembers the first $|y| - 1$ symbols and is able to check
that for no nonempty prefix $x_2'$ of $x_2$, $y_3 x_2'$
is a suffix of $y_2 y_3$.

If the checks in both (i) and (ii) are successful and at the end
the simulation of $A$ ends with a final state, this means that
the decomposition $x_1 y_1 y_2 y_3 x_2$ gives a maximal site-directed
insertion of $y$ into  a string of $L(A)$.
\qed
\end{proof}

\subsection{Decision problems for maximal/minimal SDI}
\label{decisionSection}

From  Proposition~\ref{patta1} we know that the maximal or minimal SDI
of regular languages need not be regular. However, for regular
languages $L_1$ and $L_2$ we can decide membership in
$L_1 \stackrel{\rm max-sdi}{\leftarrow} L_2$
(or $L_1 \stackrel{\rm min-sdi}{\leftarrow} L_2$) in polynomial time.

\begin{lemma}
	\label{latta3}
For DFAs $A$ and $B$ and $w \in \Sigma^*$ we can decide
in  time $O(n^6)$ whether 
$w \in L(A) \stackrel{\rm max-sdi}{\leftarrow} L(B)$,
or whether $w \in L(A) \stackrel{\rm min-sdi}{\leftarrow} L(B)$.
\end{lemma}

As we have seen, the maximal and minimal SDI operations 
are often more difficult to handle than the unrestricted SDI.
Using Lemma~\ref{latta21}~(ii) we note that deciding
maximal (or minimal) SDI-freeness is the same as
deciding SDI-freeness and by Theorem~\ref{tatta20} we have:

\begin{corollary}
	\label{catta20}
For NFAs $A$ and $B$ we can decide in polynomial time
whether or not $L(A)$ is maximal SDI-free  (respectively,
minimal SDI-free)
with respect to $L(B)$.
\end{corollary}

Also, deciding whether regular languages are 
max-SDI (or min-SDI) independent
can be done in polynomial time.

\begin{theorem}
\label{tatta3}
For NFAs $A$ and $B$, we can decide in polynomial time
whether or not $L(A)$ is maximal SDI-independent
(respectively, minimal SDI-independent)
with respect to $L(B)$.
\end{theorem}

\begin{proof}
Let $\Sigma$ be the underlying alphabet of $A$ and $B$.
We verify that
$L(A) \stackrel{\rm max-sdi}{\leftarrow} \Sigma^+
= L(A) \stackrel{\rm sdi}{\leftarrow} \Sigma^+$. The inclusion
from left to right holds by Lemma~\ref{latta21}~(i).
Conversely, suppose
$w \in L(A) \stackrel{\rm sdi}{\leftarrow} y_1 y_2 y_3$, 
where $w = x_1 y_1 y_2 y_3 x_2$, $y_1, y_3 \neq \varepsilon$,
$x_1 y_1 y_3 x_2 \in L(A)$.
Then $w \in L(A) \stackrel{\rm max-sdi}{\leftarrow} x_1 y_1 y_2 y_3 x_2$,
where the latter insertion uses the outfix $(x_1 y_1, y_3 x_2)$
as insertion guide. The insertion is maximal because the outfix cannot
be expanded.
In the same way we see that
$L(A) \stackrel{\rm min-sdi}{\leftarrow} \Sigma^+
= L(A) \stackrel{\rm sdi}{\leftarrow} \Sigma^+$.
Now the claim follows by Theorem~\ref{tatta20}.
\qed
\end{proof}

Since the max-SDI and min-SDI operations do not preserve regularity
there is no straightforward algorithm to decide whether a
regular language is closed under
maximal SDI or
under minimal SDI. We conjecture that
the problem is decidable.

\begin{problem}
Is there an algorithm that for a given regular language $L$ decides
whether or not
$L \stackrel{\rm max-sdi}{\leftarrow} L \subseteq L$
(respectively,
$L \stackrel{\rm min-sdi}{\leftarrow} L \subseteq L$)?
\end{problem}

Using Proposition~\ref{patta2} we 
can decide closure of a regular language under max/min-SDI
with a finite language.

\begin{corollary}
	\label{catta71}
Given a regular language $R$ and a finite language $F$ we can
decide whether or not (i) $R \stackrel{\rm max-sdi}{\leftarrow} F
\subseteq R$, (ii) 
$R \stackrel{\rm min-sdi}{\leftarrow} F
\subseteq R$. If $R$ is specified by a DFA 
and the length of the longest
string in $F$ is bounded by a constant, the algorithm works in
polynomial time.
\end{corollary}

\begin{proof}
	By Proposition~\ref{patta2} the languages
	$R_{\rm max} = R \stackrel{\rm max-sdi}{\leftarrow} F$
	 and  $R_{\rm min} = 
	R \stackrel{\rm min-sdi}{\leftarrow} F$
 are effectively regular.

Suppose $R = L(A)$ where $A$ is a DFA with $m$ states
and underlying alphabet $\Sigma$ and the length
of the longest string in $F$ is $c_F$. The NFA $B$ constructed
in the proof of Proposition~\ref{patta2} for $R_{\rm max}$ (or
$R_{\rm min}$) has $O( m \cdot |\Sigma|^{c_F})$ states. Recall
that the NFA stores in the state a sequence of symbols
having length of the inserted string. Strictly speaking, 
the proof of Proposition~\ref{patta2} assumes that $F$ consists of
a single string, but a similar construction works for a finite language.
When $c_F$ is a constant, the size of $B$ is polynomial in $m$
and we can decide in polynomial time whether or not
$L(B) \cap \overline{L(A)} = \emptyset$.
	\qed
\end{proof}

The max-SDI and min-SDI operations do not preserve regularity and,
consequently, they cannot be represented using semantic shuffle
on trajectories. Thus, the tools developed in
Section~6 of \cite{Domaratzki2004}
to deal with language equations are not available and it remains
open whether we can solve language equations involving max-SDI or min-SDI.

\begin{problem}
Let $L$ and $R$ be regular languages. Is it decidable whether 
the equation
$X \stackrel{\rm max-sdi}{\leftarrow} L = R$
(respectively, $L \stackrel{\rm max-sdi}{\leftarrow} X = R$,
$X \stackrel{\rm min-sdi}{\leftarrow} L = R$,
$L \stackrel{\rm min-sdi}{\leftarrow} X = R$) has a solution?
\end{problem}

\section{Nondeterministic state complexity}
\label{nondetSection}

The site-directed insertion (SDI)
operation preserves regularity~\cite{Cho2017} (above stated as
Proposition~\ref{patta31})
and the construction
can be modified to show that also alphabetic SDI preserves regularity.
To conclude,  we  consider the nondeterministic state complexity 
of these operations.

\begin{lemma}
	\label{latta1}
For NFAs $M$ and $N$ having, respectively, $m$ and $n$ states,
the language $L(M) \stackrel{\rm a-sdi}{\leftarrow} L(N)$ can be
recognized by an NFA with $mn + 2m$ states.
\end{lemma}

The upper bound is the same as the bound for the nondeterministic
state complexity of ordinary insertion \cite{Han2016}, however, the
construction used  for Lemma~\ref{latta1} is not the same.
Using 
Lemma~\ref{latta23}  (the fooling set lemma \cite{Birget1992})
we can establish a matching lower bound. 

\begin{lemma}
	\label{latta2}
	For $m, n \in \mathbb{N}$, there
 exist regular languages $L_1$ and $L_2$ over a binary alphabet 
having NFAs with $m$ and $n$ states, respectively, such that
any NFA for $ L_1 \stackrel{\rm a-sdi}{\leftarrow} L_2$ needs
at least $mn + 2m$ states.
\end{lemma}

The above lemmas establish  the precise nondeterministic state complexity
 of alphabetic SDI. 

\begin{corollary}
	\label{catta1}
The worst case nondeterministic state complexity of the alphabetic
site-directed insertion of an $n$-state NFA language into an $m$-state
NFA language is $mn + 2m$. The lower bound can be reached by
languages over a binary alphabet.
\end{corollary}

It is less obvious what 
is the precise nondeterministic state complexity of the general
SDI.
If $A$ has $m$ states and $B$ has $n$ states,
Proposition~\ref{patta31}
gives an upper bound 
$3mn + 2m$ for the nondeterministic state complexity
of
 	$L(A) \stackrel{\rm sdi}{\leftarrow} L(B)$.
	Likely the bound cannot be improved  but we do not
	have a proof for the lower bound.

\begin{problem}
	\label{OGInsc} 
	What is the nondeterministic state complexity
	of site-directed insertion?
\end{problem}

\end{document}